\documentclass[runningheads]{llncs}
\usepackage{graphicx}
\usepackage{amssymb,amscd,amsfonts,amsmath}
\usepackage{tikz}
\newcommand{\rmv}[1]{}

\begin{document}
\title{Bounds on Covering Codes in RT spaces using Ordered Covering Arrays}
%
%
\author{Andr\'{e} Guerino Castoldi\inst{1}\orcidID{0000-0002-8601-4715} \and Emerson Luiz do Monte Carmelo\inst{2}\orcidID{0000-0002-5390-6901} \and Lucia Moura\inst{3}\orcidID{0000-0003-1763-2584} \and Daniel Panario\inst{4}\orcidID{ 0000-0003-3551-4063}  \and Brett Stevens\inst{4}\orcidID{0000-0003-4336-1773}}
\authorrunning{A.G. Castoldi, E.L. Monte Carmelo, L. Moura, D. Panario and B. Stevens}
%
\institute{{1} Departamento Acad\^{e}mico de Matem\'{a}tica,
Universidade Tecnol\'{o}gica Federal do Paran\'{a}, Pato Branco, Brazil,
\email{andrecastoldi@utfpr.edu.br}\\
{2} Departamento de Matem\'{a}tica, Universidade Estadual de Maring\'{a}, Brazil,
\email{elmcarmelo@uem.br}\\
{3} School of Electrical Engineering and Computer Science, University of Ottawa, Canada,
\email{lmoura@uottawa.ca}\\
{4} School of Mathematics and Statistics, Carleton University, Canada,
\email{\{daniel,brett\}@math.carleton.ca}
}
\maketitle              

\makebox[\linewidth]{\small April 2019}

\begin{abstract}
In this work, constructions of ordered covering arrays are discussed and applied to obtain new upper bounds on covering codes in Rosenbloom-Tsfasman spaces (RT spaces), improving or extending some previous results.

\end{abstract}

\section{Introduction}

Roughly speaking, covering codes deal with the following problem: Given a metric space, how many balls are enough to cover all the space?
Several applications, such as data transmission, cellular telecommunications, decoding of errors, football pool problem, have motivated the study of covering codes in Hamming spaces. Covering codes also have connections with other branches of mathematics and computer science, such as finite fields, linear algebra, graph theory, combinatorial optimization, mathematical programming, and metaheuristic search. We refer the reader to the book by Cohen et al.~\cite{cohen1997covering} for an overview of the topic.

Rosenbloom and Tsfasman \cite{rosenbloom1997codes} introduced the RT metric on
linear spaces over finite fields, motivated by possible applications
to interference in parallel channels of communication systems.
Since the RT metric generalizes the Hamming metric, central concepts on codes in
Hamming spaces have been investigated in RT space, like perfect
codes, MDS codes, linear codes, distribution, packing and covering problems.

Most research on codes using  the RT metric focuses on packing codes; covering codes in RT spaces have not been much explored. Brualdi et al.~\cite{brualdi1995codes} implicitly investigated such codes  when the space is induced by a chain. The same class is mentioned as a function in \cite{yildiz2010covering}.  An extension to an arbitrary RT space is proposed in \cite{castoldi2015covering}, which deals mainly with upper bounds, inductive relations and some sharp bounds as well as relations with MDS codes.
More recently, \cite{castoldi2018partial} improved the sphere covering bound in RT spaces under some conditions by generalizing the excess counting method.
In this work, we explore upper bounds and inductive relations for covering codes in RT spaces by using ordered covering arrays (OCA), as briefly described below.

Ordered covering arrays (OCAs) are a generalization of ordered orthogonal arrays (OOA) and covering arrays (CA).  Orthogonal arrays are classical combinatorial designs with close connections to coding theory; see the book on orthogonal arrays by Hedayat et al.~\cite{hedayat2012orthogonal}. Covering arrays generalize orthogonal arrays and have been given a lot of attention due to their use in software testing and interesting connections with other combinatorial designs; see the survey paper by Colbourn~\cite{colbourn2004combinatorial}. Ordered orthogonal arrays are a generalization of orthogonal arrays introduced independently by
Lawrence \cite{lawrence1996combinatorial} and Mullen and Schmid \cite{mullen1996equivalence}, and are used in numerical integration. OCAs have been  introduced more recently by Krikorian~\cite{krikorian2011combinatorial}, generalizing several of the mentioned designs; their definition is given in Section 3.
In~\cite{krikorian2011combinatorial}, Krikorian gives recursive and Roux-type
constructions of OCAs as well as other constructions using the columns of a covering array and discusses an application of OCAs to numerical integration (evaluating multi-dimensional integrals).

In this paper, we apply OCAs to obtain new upper bounds on covering codes in RT spaces.
We review the basics on RT metric and covering codes in Section \ref{sec1}. CAs and OCAs are defined in Section \ref{sec2}. Section \ref{sec3} is devoted to recursive relations on the parameters of OCAs. Finally, in Section \ref{sec4}, we obtain upper bounds on covering codes in RT spaces from OCAs.

\section{Preliminaries: RT metric and covering codes} \label{sec1}

We review the RT metric based on \cite{brualdi1995codes}.
Let $P$ be a finite partial
ordered set (poset) and denote its partial order relation by
$\preceq$. A poset is a {\it chain} when any two elements
are comparable; a poset is an {\it anti-chain} when no two
distinct elements are comparable. A subset $I$ of $P$ is an {\it ideal} of $P$ when the following
property holds: if $b\in I$ and $a\preceq b$, then $a\in I$.
The {\it ideal generated} by a subset $A$ of $P$ is the ideal
of the smallest cardinality which contains $A$, denoted by
$\langle A\rangle $. An element $a\in I$ is {\it maximal in I}
if $a\preceq b$ implies that $b=a$. Analogously, an element
$a\in I$ is {\it minimal in I} if $b\preceq a$ implies that
$b=a$. A subset $J$ of $P$ is an {\it anti-ideal} of $P$ when it is
the complement of an ideal of $P$. If an ideal $I$ has
$t$ elements, then its corresponding anti-ideal has $n-t$
elements, where $n$ denotes the number of elements in $P$.

Let $m$ and $s$ be positive integers and $\Omega[m,s]$ be a
set of $ms$ elements partitioned into $m$ blocks $B_{i}$ having
$s$ elements each, where $B_{i}=\{b_{is},\ldots, b_{(i+1)s-1}\}$
for $i=0,\ldots,m-1$ and the elements of each block are ordered
as $b_{is} \preceq b_{is+1} \preceq \cdots \preceq b_{(i+1)s-1}$.
The set $\Omega[m,s]$ has a structure of a poset: it is the
union of $m$ disjoint chains, each one having $s$ elements,
which is known as the \emph{Rosenbloom-Tsfasman poset
$\Omega[m,s]$}, or briefly an RT poset $\Omega[m,s]$. When
$\Omega[m,s]=[m\times s]:=\{1,\ldots,ms\}$, the RT poset
$\Omega[m,s]$ is denoted by RT poset $[m\times s]$ and its
blocks are $B_{i}=\{i+1,\ldots, (i+1)s\}$, for $i=0,\ldots,m-1$.

For $1 \leq i \leq ms$ and $1 \leq j \leq \min\{m,i\},$ the
parameter $\Omega_{j}(i)$ denotes the number of ideals of the
RT poset $[m\times s]$ whose cardinality is $i$ with exactly $j$
maximal elements. In \cite[Proposition 1]{castoldi2018partial},
it is shown that $\Omega_{j}(i)={m\choose j}{i-1 \choose j-1}$
if $j\leq i \leq s$.

The {\it RT distance} between $x=(x_{1},\ldots,x_{ms})$ and
$y=(y_{1},\ldots,y_{ms})$ in $\mathbb{Z}_{q}^{ms}$ is defined
as \cite{brualdi1995codes}
$$d_{RT}(x,y)=|\langle supp(x-y)\rangle |=|\langle \{i: x_{i}\neq y_{i}\} \rangle |.$$
A set $\mathbb{Z}_{q}^{ms}$ endowed with the RT distance is a
{\it Rosenbloom-Tsfasman space}, or simply, an {\it RT space}.

The RT sphere centered at $x$ of radius $R$, denoted by
$B^{RT}(x,R)=\{y\in \mathbb{Z}_{q}^{ms}: d_{RT}(x,y)\leq R\},$
has cardinality given by the formula
\begin{align}
V_{q}^{RT}(m,s,R)=1+\sum_{i=1}^{R}\sum_{j=1}^{\min\{m,i\}}
    q^{i-j}(q-1)^{j}\Omega_{j}(i). \label{lenrt}
\end{align}
As expected, the case $s=1$ corresponds to the classical Hamming
sphere. Indeed, each subset produces an ideal formed by minimal
elements of the anti-chain $[m\times 1]$, thus the parameters
$\Omega_{i}(i)={m \choose i}$ and $\Omega_{j}(i)=0$ for $j<i$ yield
\begin{align}
\label{bola}
V_{q}(m,R)=V_{q}^{RT}(m,1,R)=1+\sum_{i=1}^{R}(q-1)^{i}{m \choose i}.
\end{align}

In contrast with the Hamming space, the computation of the sum
in Eq.~\eqref{lenrt} is not a feasible procedure for a general
RT space. In addition to the well studied case $s=1$, it is known that
$V_{q}^{RT}(1,s,R)=q^R$ for a space induced by a chain
$[1\times s]$, see \cite[Theorem 2.1]{brualdi1995codes} and
\cite{yildiz2010covering}. Also, it is proved in
\cite[Corollary 1]{castoldi2018partial} that, for $R\leq s$,
$$V_{q}^{RT}(m,s,R)=1+\sum_{i=1}^{R}\sum_{j=1}^{\min\{m,i\}}
q^{i-j}(q-1)^{j}{m\choose j}{i-1\choose j-1}.$$

We now define covering codes in an arbitrary RT space, and
refer the reader to \cite{castoldi2015covering} for an overview.

\begin{definition}
Given an RT poset $[m\times s]$, let $C$ be a subset of
$\mathbb{Z}_{q}^{ms}$. The code $C$ is an $R$-{\it covering} of
the RT space $\mathbb{Z}_{q}^{ms}$ if
for every $x\in \mathbb{Z}_{q}^{ms}$ there is a codeword
$c\in C$ such that $d_{RT}(x,c)\leq R$, or equivalently,
\[
\bigcup_{c\in C}B^{RT}(c,R)=\mathbb{Z}_{q}^{ms}.
\]
The number $K_{q}^{RT}(m,s,R)$ denotes the smallest cardinality
of an $R$-covering of the RT space $\mathbb{Z}_{q}^{ms}$.
\end{definition}

In particular, $K_{q}^{RT}(m,1,R)=K_{q}(m,R).$  The sphere
covering bound and a general upper bound are stated below.

\begin{proposition}(\cite[Propositions 6 and 7]{castoldi2015covering})\label{triviais}
For every $q\geq 2$ and $R$ such that $0<R<ms$,
$$\dfrac{q^{ms}}{V_{q}^{RT}(m,s,R)} \leq K_{q}^{RT}(m,s,R)
\leq q^{ms-R}.$$
\end{proposition}

\section{Ordered covering arrays}\label{sec2}

In this section, we define an important combinatorial object for this paper. We start recalling two classical combinatorial structures.

\begin{definition}
Let $t$,  $v$, $\lambda$, $n$, $N$ be positive integers
 and $N\geq\lambda v^{t}$. Let $A$
be an $N\times n$ array over an alphabet $V$ of size $v$. An
$N\times t$ subarray of $A$ is {\bf $\lambda$-covered} if it
has each $t$-tuple over $V$ as a row at least $\lambda$ times.
A set of $t$ columns of $A$ is {\bf $\lambda$-covered} if the $N\times t$
subarray of $A$ formed by them is $\lambda$-covered; when $\lambda = 1$ we simply say it is covered.
\end{definition}

In what follows, whenever $\lambda=1$, we omit $\lambda$ from the notation.

\begin{definition}(CA and OA)
Let $N$, $n$, $v$ and $\lambda$ be positive integers such that $2\leq t \leq n$. A  {\bf covering array} $CA_{\lambda}(N;t,n,v)$ is an $N\times n$ array $A$ with entries from a set $V$ of size $v$  such that any $t$-set of columns of $A$ is $\lambda$-covered. The parameter $t$ is the {\bf strength} of the covering array. The {\bf covering array number}  $CAN_{\lambda}(t,n,v)$ is the smallest positive integer $N$ such that a $CA_{\lambda}(N;t,n,v)$ exists. An {\bf orthogonal array} is a covering array with $N=\lambda v^t$.
\end{definition}

We are now ready to introduce ordered covering arrays.

\begin{definition} (OCA and OOA)
Let $t$, $m$, $s$, $v$ and $\lambda$ be positive integers such
that $2\leq t \leq ms$. An {\bf ordered covering array}
$OCA_{\lambda}(N;t,m,s,v)$ is an $N\times ms$ array $A$ with
entries from an alphabet $V$ of size $v$, whose columns are
labeled by an RT poset $\Omega[m,s]$, satisfying the property:
for each anti-ideal $J$ of the RT poset $\Omega[m,s]$ with $|J|=t$,
the set of columns of $A$ labeled by $J$ is $\lambda$-covered.
The parameter $t$ is the {\bf strength} of the ordered covering array.
The {\bf ordered covering array number}  $OCAN_{\lambda}(t,m,s,v)$  is the
smallest positive integer $N$ such that there exists an $OCA_{\lambda}(N;t,m,s,v)$.
An {\bf ordered orthogonal array} is an ordered covering array with $N=\lambda v^t$.
\end{definition}

\begin{remark}
Ordered covering arrays are special cases of variable strength covering arrays \cite{raaphorst2013variable,raaphorst2018variable}.
Ordered covering arrays were first studied by Krikorian~\cite{krikorian2011combinatorial}.
\end{remark}

\begin{example}\label{ocaexample}
The following array is an OCA of strength 2 with 5 rows:
\[
OCA(5;2,4,2,2)=
\begin{array}{c}
\begin{array}{cc|cc|cc|cc} \ 1 & 2 & 3 & 4 & 5 & 6 & 7 & 8\hspace{1.5mm}  \ \\
\end{array}\\
\left[
\begin{array}{cc|cc|cc|cc}
0 & 1 & 0 & 1 & 0 & 1 & 0 & 1  \\
1 & 1 & 1 & 0 & 0 & 0 & 0 & 0  \\
0 & 0 & 1 & 1 & 1 & 0 & 1 & 0  \\
1 & 0 & 0 & 0 & 1 & 1 & 0 & 0  \\
0 & 0 & 0 & 0 & 0 & 0 & 1 & 1
\end{array}
\right].
\end{array}
\]
The columns of this array are labeled by $[4\times 2]=\{1,\ldots,8\}$
and the blocks of the RT poset $[4 \times 2]$ are $B_{0}=\{1,2\}$,
$B_{1}=\{3,4\}$, $B_{2}=\{5,6\}$ are $B_{3}=\{7,8\}$. We have ten
anti-ideals of size 2, namely,
$$ \{1, 2\}, \{3, 4\}, \{5, 6\}, \{7, 8\}, \{2, 4\},
   \{2, 6\}, \{2, 8\}, \{4, 6\}, \{4, 8\}, \{6, 8\}.$$
The $5\times 2$ subarray constructed from each one of
theses anti-ideals covers all the pairs $(0,0)$, $(0,1)$, $(1,0)$
and $(1,1)$ at least once.
\end{example}

In an $OCA_{\lambda}(N;t,m,s,v)$ such that
$s>t$, each one of the first $s-t$ elements of a block in the RT
poset $\Omega[m,s]$ is not an element of any anti-ideal of size $t$.
 Therefore, a column labeled by one of these elements will not be part of any $N\times t$ subarray that must be $\lambda$-covered in an OCA.
 So we assume $s\leq t$ from now on.

 Two trivial relationships between
the ordered covering array number and the covering array
number $CAN_{\lambda}(t,n,v)$ are:
\begin{enumerate}
\item[(1)] $\lambda v^{t}\leq OCAN_{\lambda}(t,m,s,v) \leq
CAN_{\lambda}(t,ms,v) $;
\item[(2)] For $t\leq m$, $CAN_{\lambda}(t,m,v)\leq
OCAN_{\lambda}(t,m,t,v)$.
\end{enumerate}

If $\lambda=1$, we just write $OCA(N;t,m,s,v)$. We observe that if $N=\lambda v^t$,
an $OCA_{\lambda}(\lambda v^{t};t,m,s,v)$ is an ordered orthogonal
array $OOA_{\lambda}(\lambda v^{t};t,m,s,v)$. When $s=1$, an
$OCA_{\lambda}(N;t,m,1,v)$ is the well-known covering array
$CA_{\lambda}(N;t,m,v)$.

\section{Recursive relations for ordered covering arrays}\label{sec3}

In this section, we show recursive relations for ordered coverings arrays.

\begin{proposition} \label{prop8}
\begin{enumerate}
\item[(1)] If there exists an $OCA_{\lambda}(N;t,m,t-1,v)$,
then there exists an $OCA_{\lambda}(N;t,m,t,v)$.
\item[(2)] If  there exists an
$OCA_{\lambda}(N;t,m,s,v)$, then there exists an
$OCA_{\lambda}(N;t,m,s-1,v)$.
\item[(3)] If  there exists an
$OCA_{\lambda}(N;t,m,s,v)$, then there exists an
$OCA_{\lambda}(N;t,m-1,s,v)$.
\end{enumerate}
\end{proposition}

\begin{proof}
Parts (2) and (3) are easily obtained by deletion of appropriate columns of the $OCA_{\lambda}(N;t,m,s,v)$.
We prove part (1). Let $P$ be an  RT poset $\Omega[m,t]$, and let
$B_0, B_1, \ldots, B_{m-1}$ be the blocks of $P$. Each $B_i$, $0\leq i <m$,   is a chain, and we
denote by $\min(B_i)$ and $\max(B_i)$ the minimum and maximum elements of $B_i$, respectively.
Let $M=\{\max(B_0), \ldots, \max(B_{m-1})\}$ and let $\pi$ be a derangement of $M$. Let $P'=P\setminus\{\min(B_0), \ldots, \min(B_{m-1})\}$,
an RT poset $\Omega[m,t-1]$,  and let $A'$ be an $OCA_{\lambda}(N;t,m,t-1,v)$ with columns labeled by $P'$. We take a map $f\colon P \rightarrow P'$
given by $f(x)=x$ if $x\in P'$ and $f(x)=\pi(\max(B_i))$, if $x=\min(B_i)$.
In Fig.~\ref{fg4}, we depict $P'$ (above) and  $P$ (below) where $\min(B_i)$ is labeled by $\overline{a}$ where $a=\pi(\max(B_i))$.
Construct an array $A$ with columns labeled by elements of $P$ by taking the column  of $A$ labeled by $x$ to be the column  of $A'$ labeled by $f(x)$.
Let $J$ be any anti-ideal of $P$ of size $t$ and let  $J'=f(J)$. Then either $J=J'\subseteq P'$, or $\min(B_i)\in J$ for some $i$, which implies $J=B_i$ and $J'=(B_i \setminus\{\min(B_i)\})\cup \max(B_j)$ for some $j\not=i$. In either case, $J'$ is an anti-ideal of $P'$, and so the set of $t$ columns of $A'$ corresponding to $J'$ is $\lambda$-covered. Therefore, the set of $t$ columns of $A$ corresponding to $J$ is  $\lambda$-covered, for any anti-ideal $J$ of $P$ of size $t$, and $A$ is an $OCA_{\lambda}(N;t,m,t,v)$. \hfill $\Box$
\end{proof}
\begin{figure}[t]
\begin{tabular}{|ccc|} \hline
\ \ \ \ &
    \begin{tikzpicture}[scale=.9]
    \draw[fill] (-1,0) circle (.05cm) node[right] {$1$};
    \draw[fill] (-1,1) circle (.05cm) node[right] {$2$};
    \draw[fill] (-1,2) circle (.05cm) node[right] {$t-1$};
    \draw[fill] (1,0) circle (.05cm) node[right] {$t$};
    \draw[fill] (1,1) circle (.05cm) node[right] {$t+1$};
    \draw[fill] (1,2) circle (.05cm) node[right] {$2(t-1)$};
    \draw[fill] (5,0) circle (.05cm) node[right] {$(m-2)(t-1)+1$};
    \draw[fill] (5,1) circle (.05cm) node[right] {$(m-2)(t-1)+2$};
    \draw[fill] (5,2) circle (.05cm) node[right] {$(m-1)(t-1)$};
    \draw[fill] (8.5,0) circle (.05cm) node[right] {$(m-1)(t-1)+1$};
    \draw[fill] (8.5,1) circle (.05cm) node[right] {$(m-1)(t-1)+2$};
    \draw[fill] (8.5,2) circle (.05cm) node[right] {$m(t-1)$};
    \draw[fill] (3,1) circle (.02cm) {};
    \draw[fill] (3.25,1) circle (.02cm) {};
    \draw[fill] (3.5,1) circle (.02cm) {};
      \draw (-1,0) -- (-1,1);
      \draw[dotted] (-1,1) -- (-1,2);
      \draw (1,0) -- (1,1);
      \draw[dotted] (1,1) -- (1,2);
      \draw (5,0) -- (5,1);
      \draw[dotted] (5,1) -- (5,2);
      \draw (8.5,0) -- (8.5,1);
      \draw[dotted] (8.5,1) -- (8.5,2);
    \end{tikzpicture} & \ \ \ \\
    \hline
&
    \begin{tikzpicture}[scale=.9]
    \draw[fill] (-1,0) circle (.05cm) node[right] {$\overline{m(t-1)}$};
    \draw[fill] (-1,1) circle (.05cm) node[right] {$1$};
    \draw[fill] (-1,2) circle (.05cm) node[right] {$t-1$};
    \draw[fill] (1,0) circle (.05cm) node[right] {$\overline{t-1}$};
    \draw[fill] (1,1) circle (.05cm) node[right] {$t$};
    \draw[fill] (1,2) circle (.05cm) node[right] {$2(t-1)$};
    \draw[fill] (5,0) circle (.05cm) node[right] {$\overline{(m-2)(t-1)}$};
    \draw[fill] (5,1) circle (.05cm) node[right] {$(m-2)(t-1)+1$};
    \draw[fill] (5,2) circle (.05cm) node[right] {$(m-1)(t-1)$};
    \draw[fill] (8.5,0) circle (.05cm) node[right] {$\overline{(m-1)(t-1)}$};
    \draw[fill] (8.5,1) circle (.05cm) node[right] {$(m-1)(t-1)+1$};
    \draw[fill] (8.5,2) circle (.05cm) node[right] {$m(t-1)$};
    \draw[fill] (3,1) circle (.02cm) {};
    \draw[fill] (3.25,1) circle (.02cm) {};
    \draw[fill] (3.5,1) circle (.02cm) {};
      \draw (-1,0) -- (-1,1);
      \draw[dotted] (-1,1) -- (-1,2);
      \draw (1,0) -- (1,1);
      \draw[dotted] (1,1) -- (1,2);
      \draw (5,0) -- (5,1);
      \draw[dotted] (5,1) -- (5,2);
      \draw (8.5,0) -- (8.5,1);
      \draw[dotted] (8.5,1) -- (8.5,2);
    \end{tikzpicture} &\\ \hline
  \end{tabular}
\caption{Above, $P'$  is an RT poset $\Omega[m,t-1]$;  below, $P\supset P'$ is an RT poset $\Omega[m,t]$.}
\label{fg4}
\end{figure}
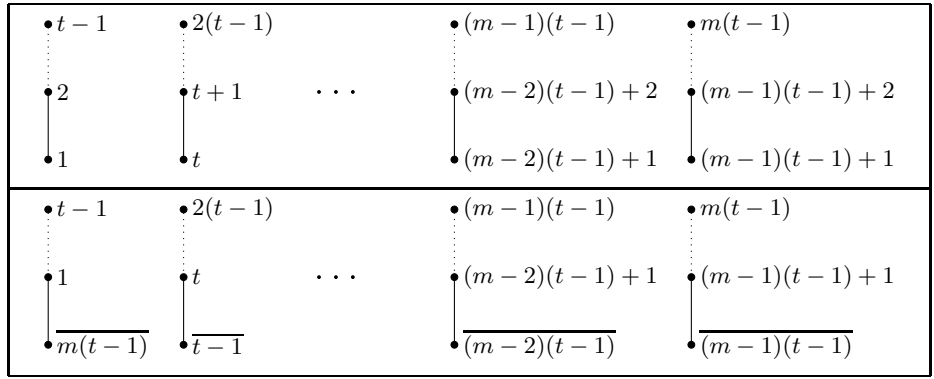

As a straightforward  consequence of Proposition \ref{prop8} items $(1)$ and $(2)$,  we have the following result.

\begin{corollary}
There exists an  $OCA_{\lambda}(N;t,m,t,v)$ if and only if there exists an $OCA_{\lambda}(N;t,m,t-1,v)$.
\end{corollary}

By the above corollary, when $t=2$, the right hand side OCA has $s=t-1=1$ which corresponds to a covering array; therefore, in this case,
the ordered covering array number is equal to the covering array number. This also shows
that we need $t>2$ in order to have ordering covering arrays
essentially different than covering arrays.

\begin{example}\label{thm6}
$OCAN_{\lambda}(2,m,2,v)=CAN_{\lambda}(2,m,v).$

Let us label the columns of a $CA_{\lambda}(N;2,m,v)$ by the elements of $[m]=\{1,\ldots,m\}$.
Choose the columns of $CA_{\lambda}(N;2,m,v)$ labeled by the elements of $[m]$ given in Fig~\ref{fg1}.
We use again the notation $\overline{a}$ to duplicate $a \in[m]$ in the RT
poset $\Omega[m,2]$ in such a way that $a$ and $\overline{a}$ are not comparable, but the columns of $CA_{\lambda}(N;2,m,v)$ labeled by $a$ and $\overline{a}$ are equal.
\begin{figure}[h]
\centering
    \begin{tikzpicture}[scale=.9]
    \draw[fill] (-1,1) circle (.05cm) node[right] {$\overline{2}$};
    \draw[fill] (-1,2) circle (.05cm) node[right] {$1$};
    \draw[fill] (1,1) circle (.05cm) node[right] {$\overline{3}$};
    \draw[fill] (1,2) circle (.05cm) node[right] {$2$};
    \draw[fill] (3.5,1) circle (.05cm) node[right] {$\overline{m}$};
    \draw[fill] (3.5,2) circle (.05cm) node[right] {$m-1$};
    \draw[fill] (5.5,1) circle (.05cm) node[right] {$\overline{1}$};
    \draw[fill] (5.5,2) circle (.05cm) node[right] {$m$};
    \draw[fill] (2,1.5) circle (.02cm) {};
    \draw[fill] (2.25,1.5) circle (.02cm) {};
    \draw[fill] (2.5,1.5) circle (.02cm) {};
      \draw (-1,1) -- (-1,2);
      \draw (1,1) -- (1,2);
      \draw (3.5,1) -- (3.5,2);
      \draw (5.5,1) -- (5.5,2);
    \end{tikzpicture}
    \caption{Blocks of the RT poset $\Omega[m,2]$.}
    \label{fg1}
\end{figure}
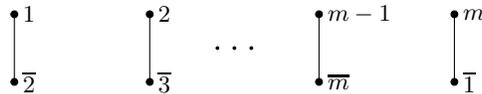

\end{example}

It is known from \cite{kleitman1973families} that $CAN(2,m,2)$
is the smallest positive integer $N$ such that $m\leq {N-1
\choose \lfloor \frac{N}{2}\rfloor -1}$. As a consequence,
 we obtain the following result.
\begin{corollary}
The ordered covering array number $OCAN(2,m,2,2)$ is
the smallest positive integer $N$ such that
$m\leq {N-1 \choose \lfloor \frac{N}{2}\rfloor -1}$.
\end{corollary}

In the next result, we show an upper bound for the ordered covering
array number over an alphabet of size $v$ from the ordered covering
array number over an alphabet of size $v+1$. It generalizes
\cite[Lemma 3.1]{colbourn2008strength} and part of
\cite[Lemma 3.1]{colbourn2010covering}.

\begin{theorem}\label{fusion}
(Fusion Theorem) $OCAN_{\lambda}(t,m,s,v)\leq
OCAN_{\lambda}(t,m,s,v+1)-2$.
\end{theorem}

\begin{proof}
Let $V=\{1,\ldots, v+1\}$ be an alphabet of size $v+1$ and consider an $OCA_{\lambda}(N;t,m,s,v+1)$ over $V$.
The permutation of the entries of any column of the $OCA_{\lambda}(N;t,m,s,v+1)$ still produces an ordered covering array with the same parameters.
If necessary, applying a permutation in each of the columns, we can guarantee that there exists a row in $OCA_{\lambda}(N;t,m,s,v+1)$
such that all the entries are $v+1$. We delete this row.

Choose a second row $r=(c_{1},\ldots,c_{ms})$ of $OCA_{\lambda}(N;t,m,s,v+1)$. In every row except $r$, where there exist an entry $v+1$ in column
$i$, replace $v+1$ by $c_{i}$ if $c_{i}\neq v+1$, otherwise, replace
$v+1$ by any element of $\{1,\ldots,v\}$. Delete row $r$.

The array $A$ obtained by deleting these two rows
of $OCA_{\lambda}(N;t,m,s,v+1)$ is an $OCA_{\lambda}(N-2;t,m,s,v)$, since each $t$-tuple that was covered by row $r$ is now covered by one or more of the modified rows.
\hfill $\Box$
\end{proof}

By \cite[Theorem 3]{castoldi2017ordered} and Theorem \ref{fusion}, we derive an upper bound on the ordered covering array number.

\begin{corollary}\label{corol3}
Let $q$ be a prime power. Then, $OCAN(t,q+1,t,q-1)\leq q^{t}-2$.
\end{corollary}

\section{Constructions of covering codes using covering arrays}\label{sec4}

In this section, ordered covering arrays are used to construct
covering codes in RT spaces yielding upper bounds on their size. Theorems \ref{teo1}
and \ref{teoca}  are a generalization of results
already discovered for covering codes in Hamming spaces connected
with surjective matrices \cite{cohen1997covering}.

Let $I=\{i_{1},\ldots,i_{k}\}$ be a subset of $[n]=\{1,\ldots,n\}$.
Given an element $x=(x_{1},\ldots,x_{n})\in \mathbb{Z}_{q}^{n},$
the projection of $x$ with respect to $I$ is the element
$\pi_{I}(x)=(x_{i_{1}},\ldots,x_{i_{k}})\in \mathbb{Z}_{q}^{k}$.
More generally, for a non-empty subset $C$ of $\mathbb{Z}_{q}^{n}$,
the projection of $C$ with respect to $I$ is the set
$\pi_{I}(C)=\{\pi_{I}(c): c \in C\}$.

In \cite[Theorem 13]{castoldi2015covering}, it was proved that
$K_{q}^{RT}(m,s,ms-t)=q$ if  \linebreak $m\geq (t-1)q+1$. What can we say
about $K_{q}^{RT}(m,s,ms-t)$ when $m=(t-1)q$? In this direction,
we have the following result. A reference for item $(1)$ is
\cite{cohen1997covering}; the other two items are original
results of the present paper.

\begin{theorem}\label{teo1}
For $t\geq 2$,
\begin{enumerate}
\item[(1)] $K_{q}((t-1)q,(t-1)q-t)\leq q-2+CAN(t,(t-1)q,2).$
\item[(2)] $K_{q}^{RT}((t-1)q,s,(t-1)qs-t)\leq K_{q}((t-1)q,(t-1)q-t)$.
\item[(3)] $K_{q}^{RT}((t-1)q,s,(t-1)qs-t)\leq q-2+CAN(t,(t-1)q,2).$
\end{enumerate}
\end{theorem}

\begin{proof} (Sketch.) We only prove part $(2)$ here; part
$(3)$ is straightforward from (1) and (2).
Let $M$ be the set of all maximal elements of the RT poset $[(t-1)q\times s]$ and
$C'$ be a $((t-1)q-t)$-covering code of the Hamming space $\mathbb{Z}_{q}^{(t-1)q}$.
Let $C$ be the subset of $\mathbb{Z}_{q}^{(t-1)qs}$ such that $c\in C$ if and only if
$\pi_{M}(c) \in C'$ and all the other coordinates of $c$ are equal to zero.
Given $x\in \mathbb{Z}_{q}^{(t-1)qs}$, let $\pi_{M}(x)\in \mathbb{Z}_{q}^{(t-1)q}$. Since
$C'$ is a $((t-1)q-t)$-covering of the Hamming space $\mathbb{Z}_{q}^{(t-1)q}$,
there exists $c'\in C'$  such that $\pi_{M}(x)$ and $c'$ coincide in at least $t$ coordinates.
Let $c\in C$ such that $\pi_{M}(c)=c'$.
Therefore $d_{RT}(x,c)\leq  (t-1)qs-t$, and $C$ is a $((t-1)qs-t)$-covering of the RT space $\mathbb{Z}_{q}^{(t-1)qs-t}$.
\hfill $\Box$
\end{proof}

Applying the trivial bounds we have that
$3\leq K_{3}^{RT}(3,s,3s-2)\leq 9$. The upper bound can be
improved by Theorem \ref{teo1}.

\begin{corollary}
$K_{3}^{RT}(3,s,3s-2)\leq 5.$
\end{corollary}

\begin{proof}
Theorem \ref{teo1} yields $K_{3}^{RT}(3,s,3s-2)\leq 1+ CAN(2,3,2)$.
On the other hand, $CAN(2,3,2)=4$, according to
\cite{kleitman1973families}, and the upper bound follows.
\hfill $\Box$
\end{proof}

MDS codes have been used to improve upper bounds on $K_{q}(n,R)$
\cite{blokhuis1984more,carnielli1985covering,cohen1997covering}.
In \cite[Theorem 30]{castoldi2015covering}, MDS codes in RT spaces
are used to improve upper bounds for $K_{q}^{RT}(m,s,R)$. We
generalize these results using ordered covering arrays.

\begin{theorem}\label{teoca}
$K_{vq}^{RT}(m,s,R)\leq OCAN(ms-R,m,s,v) K_{q}^{RT}(m,s,R).$
\end{theorem}

\begin{proof}
Throughout this proof, the set $\mathbb{Z}_{vq}$ is regarded as the
set $\mathbb{Z}_{vq}=\mathbb{Z}_{v}\times \mathbb{Z}_{q}$ by setting the
bijection $xq+y \rightarrow (x,y)$. This strategy allows us to
analyze the information on the coordinates $x$ and $y$ separately.

Let $H$ be an $R$-covering of the RT space $\mathbb{Z}_{q}^{ms}$, and let
$C$ be the set of the rows of an $OCA(N;ms-R,m,s,v)$. We show that
$$G=\{((c_{1},h_{1}),\ldots,(c_{ms},h_{ms}))\in \mathbb{Z}_{vq}^{ms}:(c_{1},\ldots,c_{ms}) \in C, \ (h_{1},\ldots,h_{ms}) \in H\}$$
is a $R$-covering of the RT space $\mathbb{Z}_{vq}^{ms}$.

Indeed, for
$z=((x_{1},y_{1}),\ldots,(x_{ms},y_{ms}))\in \mathbb{Z}_{vq}^{ms}$,
let $x=(x_{1},\ldots,x_{ms})$ in $\mathbb{Z}_{v}^{ms}$ and $y=(y_{1},\ldots,y_{ms})$ in $\mathbb{Z}_{q}^{ms}$.
Since $H$ is an $R$-covering of the RT space $\mathbb{Z}_{q}^{ms}$,
for $y \in \mathbb{Z}_{q}^{ms}$ there exists $h=(h_{1},\ldots,h_{ms})
\in H$ such that $d_{RT}(y,h)\leq R$. Let $I$ be the ideal generated by $supp(y-h)$ and $I'$ be an ideal of the RT poset $[m\times s]$
of size $R$ such that $I\subseteq I'$. Then there exists a codeword
$c=(c_{1},\ldots,c_{ms})$ in $C$ such that $x$ and $c$ coincide in
all coordinates of the complementary set of $I'$ (which is an anti-ideal
of size $ms-R$). Thus $supp(x-c)\subseteq I'$.

Let $g=((c_{1},h_{1}),\ldots,(c_{ms},h_{ms}))$ in $G$. By
construction, $z$ and $g$ coincide in all coordinates of the complementary set of $I'$. Therefore,
$d_{RT}(z,g)=|\langle supp(z-g)\rangle |\leq |I'|=R$ and
the proof is complete.
\hfill $\Box$
\end{proof}

Together with \cite[Theorem 13]{castoldi2015covering} we get the following consequences of Theorem \ref{teoca}.

\begin{corollary}
\begin{itemize}
\item[(1)] For $q< m \leq 2q$ and $2\leq s \leq 3$, $K_{2q}^{RT}(2m,s,2ms-3)\leq q (OCAN(3,m,s,2)+CAN(2,m,2)).$
\item[(2)] If $(t-1)q+1\leq m \leq (t-1)qv$, then $K_{qv}^{RT}(m,s,ms-t)\leq q OCAN(t,m,s,v)$.
\end{itemize}
\end{corollary}

If $OCAN(t,m,s,v)=v^{t}$, then there exists an ordered orthogonal
array $OOA(t,m,s,v)$. In the following results, ordered orthogonal
arrays are used to obtain upper bounds for covering codes in RT
spaces. Item $(1)$ has appeared in
\cite[Theorem 30]{castoldi2015covering}. Item $(2)$ is a consequence of Theorem \ref{teoca} and \cite[Theorem 3]{castoldi2017ordered}.

\begin{corollary} \label{corol4}
\begin{itemize}
\item[(1)] If there is an ordered orthogonal array $OOA(ms-R,m,s,v)$, then $K_{vq}^{RT}(m,s,R)$ $\leq v^{ms-R} K_{q}^{RT}(m,s,R).$
\item[(2)] Let $q$ be a prime power, $m\leq q+1$ and $s\leq t$. Then, we have $K_{qv}^{RT}(m,s,ms-t)$ $\leq q^{t} K_{v}^{RT}(m,s,ms-t).$
\item[(3)] Let $q$ be a prime power. For $t\geq 2$, we have $K_{(q-1)v}^{RT}(q+1,t,qt)\leq (q^t -2)K_{v}^{RT}(q+1,t,qt)$.
\end{itemize}
\end{corollary}

In order to get better upper bounds on $K_{vq}^{RT}(m,s,R)$, we
improve the upper bound on $K_{v}^{RT}(m,s,R)$ for suitable
values of $m$ and $R$. For this purpose, we look at a covering code that gives the trivial upper bound for $K_{v}^{RT}(m,s,R)$ and
modify some of its codewords to reduce the size of the covering code.

\begin{theorem}\label{twochains}
For $s\geq 2$,
\begin{itemize}
\item[(1)] $K_{v}^{RT}(2,s,s)\leq v^{s-2}(v^{2}-1)$,
\item[(2)] $K_{v}^{RT}(3,s,2s-1)\leq v(v^{s}-1)$.
\end{itemize}
\end{theorem}

\begin{proof}
$(1)$ Let $I=\{1,\ldots,s\}$ ideal of the RT poset $[2\times s]$.
The trivial upper bound for $K_{v}^{RT}(2,s,s)$ is $v^{s}$, and
a $s$-covering of the RT space $\mathbb{Z}_{v}^{2s}$ of size $v^{s}$ is
$$C=\{c\in \mathbb{Z}_{v}^{2s}: \pi_{I}(c)=0 \in \mathbb{Z}_{v}^{s}\}.$$
For each $z\in \mathbb{Z}_{v}^{s-2}$, let
$$C_{z}=\{c\in C: \pi_{J}(c)=z \in \mathbb{Z}_{v}^{s-2}\},$$
where $J=\{s+3,\ldots,2s\}$ is an anti-ideal of the RT poset $[2\times s]$. The following properties hold:
\begin{itemize}
\item[(a)] $C_{z}\cap C_{z'}=\emptyset$ if and only if $z\neq z'$;
\item[(b)] $|C_{z}|=v^{2}$ for all $z\in \mathbb{Z}_{v}^{s-2}$;
\item[(c)] $C_{z}$ is a $s$-covering of the  RT space $\mathbb{Z}_{v}^{s+2}\times \{z\}$ over the RT poset $[2\times m]$;
\item[(d)] $C=\bigcup_{z\in \mathbb{Z}_{v}^{s-2}} C_{z}$.
\end{itemize}
For each $z\in \mathbb{Z}_{v}^{s-2}$, we construct a new set $A_{z}$ from $C_{z}$
of size $v^{2}-1$ such that $A_{z}$ is a $s$-covering of the RT space $\mathbb{Z}_{v}^{s+2}\times \{z\}$.

Let $C_{z}'=C_{z}\backslash \{(0;0,0;z)\}$. For each $c=(0;c_{s+1},c_{s+2};z)\in C_{z}'$ define
\[
\phi_{z}(c)= \left\{ \begin{array}{ll}
(0;c_{s+1},c_{s+2};z) & \textrm{if $c_{s+2}=0$},\\
(0;c_{s+1},c_{s+2};c_{s+1},c_{s+2};z) & \textrm{if $c_{s+2}\neq 0$}.
\end{array}\right.
\]
We show that
$A_{z}=\{\phi_{z}(c)\in \mathbb{Z}_{v}^{2s}: c\in C_{z}'\}$
is a $s$-covering of the RT space $\mathbb{Z}_{v}^{s+2}\times \{z\}$. 
Given $x=(x_{1},\ldots,x_{s};x_{s+1},x_{s+2};z)\in \mathbb{Z}_{v}^{s+2}\times \{z\}$,
we divide the proof into three cases.
\begin{itemize}
\item[(a)] If $x_{s+2}\neq 0$, then $x$ is covered by $(0;x_{s+1},x_{s+2};x_{s+1},x_{s+2};z)$;
\item[(b)] If $x_{s+1}\neq 0$ and $x_{s+2}=0$, then $x$ is covered by $(0;x_{s+1},0;z)$;
\item[(c)] If $x_{s+1}=x_{s+2}=0$, then we have two subcases:
\begin{itemize}
\item[(i)] If $x_{s}=0$, then $x$ is covered by $(0;a,0;z)$, where $a\neq 0$;
\item[(ii)] If $x_{s}\neq 0$, then $x$ is covered by $(0;x_{s-1},x_{s};x_{s-1},x_{s};z)$.
\end{itemize}
\end{itemize}

Therefore, the set $A=\bigcup_{z\in \mathbb{Z}_{v}^{s-2}}A_{z}$ is a $s$-covering of the RT
space $\mathbb{Z}_{v}^{2s}$ of size $v^{s-2}(v^{2}-1)$.

$(2)$ Let $I=\{1,\ldots,2s-1\}$ ideal of the RT poset $[3\times s]$. The trivial upper bound for
$K_{v}^{RT}(3,s,2s-1)$ is  $v^{s+1}$, and a $(2s-1)$-covering of the RT space $\mathbb{Z}_{v}^{3s}$ is
$$C=\{c\in \mathbb{Z}_{v}^{3s}: \pi_{I}(c)=0\in \mathbb{Z}_{v}^{2s-1} \}.$$
For each $z\in \mathbb{Z}_{v}$ let  $C_{z}=\{c\in C: \pi_{2s}(c)=z \}.$
The following properties hold:
\begin{itemize}
\item[(a)] $C_{z}\cap C_{z'}=\emptyset$ if and only if $z\neq z'$;
\item[(b)] $|C_{z}|=v^{s}$ for all $z\in \mathbb{Z}_{v}$;
\item[(c)] $C_{z}$ is a $(2s-1)$-covering of the RT space $\mathbb{Z}_{v}^{2s-1}\times \{z\}\times \mathbb{Z}_{v}^{s}$
over the RT poset $[3\times s]$;
\item[(d)] $C=\bigcup_{z\in \mathbb{Z}_{v}} C_{z}$.
\end{itemize}

For each $z\in \mathbb{Z}_{v}$, we construct a new set $A_{z}$ from $C_{z}$
of size $v^{s}-1$ such that $A_{z}$ is a $(2s-1)$-covering of the RT space
$\mathbb{Z}_{v}^{2s-1}\times \{z\}\times \mathbb{Z}_{v}^{s}$.

Let $C_{z}'=C_{z}\backslash \{(0\ldots 0; 0 \ldots 0 z; 0 \ldots 0)\}$. For each $c\in C_{z}'$ define
\[
\phi(c)= \left\{ \begin{array}{ll}
c & \textrm{if $c_{3s}=0$}\\
(c_{2s+1},\ldots,c_{3s};c_{s+1},\ldots, c_{2s};c_{2s+1},\ldots,c_{3s}) & \textrm{if $c_{3s}\neq 0$}.
\end{array}\right.
\]
We claim that
$A_{z}=\{\phi_{z}(c)\in \mathbb{Z}_{v}^{3s}: c\in C_{z}'\}$
is a $(2s-1)$-covering of the RT space $\mathbb{Z}_{v}^{2s-1}\times \{z\}\times \mathbb{Z}_{v}^{s}$. Indeed,
given $x=(x_{1},\ldots,x_{2s-1},z;x_{2s+1},\ldots,x_{3s})\in $
$\mathbb{Z}_{v}^{2s-1}\times \{z\}\times \mathbb{Z}_{v}^{s}$,
we divide the proof into two cases.
\begin{itemize}
\item[(a)] If $(x_{2s+1},\ldots,x_{3s})\neq (0 \ldots 0)$, then $c\in A_{z}$ such that
$\pi_{J}(c)=(x_{2s+1},\ldots,x_{3s})$ covers $x$, where $J=\{2s+1,\ldots,3s\}$;
\item[(b)] If $(x_{2s+1},\ldots,x_{3s})= (0 \ldots 0)$, then we have two subcases:
\begin{itemize}
\item[(i)] If $x_{s}=0$, then $x$ is covered by $(0,\ldots, 0; 0, \ldots, 0, z; z', 0 ,\ldots ,0)$, where $z'\neq 0$;
\item[(ii)] If $x_{s}\neq 0$, then $x$ is covered by $(x_{1},\ldots, x_{s}; 0, \ldots, 0, z; x_{1}, \ldots, x_{s})$.
\end{itemize}
\end{itemize}

Therefore, the set $A=\bigcup_{z\in \mathbb{Z}_{v}}A_{z}$ is a
$(2s-1)$-covering of the RT space $\mathbb{Z}_{v}^{3s}$ of size
$v(v^{s}-1)$.
\hfill $\Box$
\end{proof}

\begin{example}
A $3$-covering code of the RT space $\mathbb{Z}_{2}^{6}$ (RT poset
$[2\times 3]$) with $2^3$ is $C=C_{0}\cup C_{1}$, where
\[
C_{0}=\{000000, 000010, 000100, 000110\},
\]
\[
C_{1}=\{000001, 000101, 000011, 000111\}.
\]
Theorem \ref{twochains} item $(1)$ improves the upper bound
$K_{2}^{RT}(2,3,3)\leq 8$ by using the $3$-covering code
$A=A_{0}\cup A_{1}$, where
$A_{0}=\{000100, 001100, 010110\}$, and
$A_{1}=\{ 001001, 011101, 000011\}.$
Therefore, $K_{2}^{RT}(2,3,3)\leq 6$.
\hfill $\Box$
\end{example}

\begin{corollary}\label{corol5}
Let $q$ be a prime power.
\begin{itemize}
\item[(1)] For $q+1\leq (t-1)v$, $K_{qv}^{RT}(q+1,t,qt)\leq q^{t}v^{t-2}(v^{2}-1)$.
\item[(2)] For $m\leq q+1$, $K_{qv}^{RT}(m,s,ms-(s+1))\leq q^{s+1}v(v^{s}-1)$.
\item[(3)] $K_{(q-1)v}^{RT}(q+1,t,qt)\leq (q^{t}-2)v^{t-2}(v^{2}-1)$.
\end{itemize}
\end{corollary}

\begin{proof}
$(1)$ Applying \cite[Proposition 17]{castoldi2015covering} with
$n=q-1$,
$$K_{v}^{RT}(q+1,t,qt)\leq K_{v}^{RT}(2,t,t).$$
Theorem \ref{twochains} item $(1)$ yields
$K_{v}^{RT}(q+1,t,qt)\leq v^{t-2}(v^{2}-1)$.
The result follows by Corollary \ref{corol4} item $(2)$.

$(2)$ Theorem \ref{teoca} shows that
$$K_{qv}^{RT}(m,s,ms-(s+1))\leq OCAN(s+1,m,s,q)K_{v}^{RT}(m,s,ms-(s+1)).$$
Since there exists an $OOA(s+1,q+1,s+1,q)$ then there exists an
$OOA(s+1,m,s,q)$ for $m\leq q+1$. Applying Theorem \ref{twochains}
item $(2)$ and \cite[Proposition 17]{castoldi2015covering} with
$n=m-2$,
$$K_{v}^{RT}(m,s,ms-(s+1))\leq K_{v}^{RT}(3,s,2s-1)\leq v(v^s-1).$$
Therefore, the upper bound desired is attained.

$(3)$ Applying \cite[Proposition 17]{castoldi2015covering} with
$n=q-1$,
$$K_{v}^{RT}(q+1,t,qt)\leq K_{v}^{RT}(2,t,t).$$
Theorem \ref{twochains} item $(1)$ implies that
$K_{v}^{RT}(q+1,t,qt)\leq v^{t-2}(v^2 -1)$. Corollary
\ref{corol4} item $(3)$ completes the proof.
\hfill $\Box$
\end{proof}

We compare the upper bounds for $K_{qv}^{RT}(q+1,t,qt)$. By
Theorem \ref{twochains} item $(1)$, $(qv)^{t-2}((qv)^{2}-1)$
is an upper bound for $K_{qv}^{RT}(q+1,t,qt)$. However,
Corollary \ref{corol5} item $(1)$ yields the upper bound
$q^{t}v^{t-2}(v^{2}-1)$ which improves the one given by
Theorem \ref{twochains} item $(1)$.

\bigskip
\noindent
{\bf Acknowledgements.}
The authors would like to thank the anonymous referees for their suggestions that greatly improved this paper.

A.G. Castoldi was supported by the CAPES of Brazil, Science without Borders Program, under Grant 99999.003758/2014-01.
E.L. Monte Carmelo is partially supported by CNPq/MCT grants: 311703/2016-0.
The last three authors are supported by discovery grants from NSERC of Canada.

%
%
%

\end{document}